\documentclass{article}
\usepackage[utf8]{inputenc}

    \usepackage[dvipsnames]{xcolor}
    \usepackage[T1]{fontenc}
	\usepackage[utf8]{inputenc}
	\usepackage[english]{babel}
	\usepackage{amsfonts}
	\usepackage{amssymb}
	    \usepackage{amsmath}
	
	\usepackage{enumerate}
	\usepackage{bbm}
	\usepackage{mathtools}
	\usepackage{subfig}
	\usepackage{booktabs}	
	\usepackage{setspace}
	\usepackage{graphicx}
	\usepackage{fullpage}
	
	\usepackage{soul}
	\usepackage{xspace}
	\usepackage{csquotes}
	\usepackage[margin = 1in]{geometry}
	\usepackage{placeins}
	\usepackage[footfont=normalsize,font=normalsize]{floatrow}
	\usepackage{lscape}
	\usepackage{longtable}
	\usepackage{mathabx}
    \usepackage{accents}
    \usepackage{float}

    \usepackage{tikz}
    \usetikzlibrary{calc,intersections}

	\usepackage[longnamesfirst]{natbib}
	\usepackage{bibunits}
	\defaultbibliography{Bibliography.bib}  
	\defaultbibliographystyle{aer}

	\usepackage{hyperref}
	\hypersetup{
    colorlinks=true,
    linkcolor=Red,
    filecolor=magenta,      
    urlcolor=blue,
    citecolor = Blue
}
	\usepackage{cleveref}

    \usepackage{rotating}
	\usepackage[framemethod=tikz]{mdframed}
	\usepackage{todonotes}
	\presetkeys{todonotes}{color=black!5,inline}{} 

    \usepackage{tcolorbox}
    \newtcbox{\feedback}{nobeforeafter,colframe=black,colback=white,boxrule=0.5pt,arc=2pt,
      boxsep=0pt,left=2pt,right=2pt,top=2pt,bottom=2pt,tcbox raise base}
      
    \usepackage{amsthm}

    \newtheorem{prop}{Proposition}[section]

    \theoremstyle{definition}
    \newtheorem{rem}{Remark}
    \newtheorem{defn}{Definition}
    \newtheorem{examplex}{Example}
    \newenvironment{example}
        {\pushQED{\qed}\examplex}
        {\popQED\endexamplex}
\usepackage{array}
\newcolumntype{L}[1]{>{\raggedright\let\newline\\\arraybackslash}m{#1}}
\newcolumntype{C}[1]{>{\centering\let\newline\\\arraybackslash\hspace{0pt}}m{#1}}
\newcolumntype{R}[1]{>{\raggedleft\let\newline\\\arraybackslash\hspace{0pt}}m{#1}}

\usepackage{ragged2e}
\newlength\ubwidth

\usepackage{chngcntr}

\title{An Outcome Test of Discrimination for Ranked Lists}
\author{Jonathan Roth\thanks{Brown University. \href{mailto:jonathanroth@brown.edu}{jonathanroth@brown.edu}} \and Guillaume Saint-Jacques\thanks{Apple. \href{mailto:guillaume.saintjacques@gmail.com}{guillaume.saintjacques@gmail.com}} \and YinYin Yu\thanks{LinkedIn. \href{mailto:yinyyu@linkedin.com }{yinyyu@linkedin.com}}}

\begin{document}

\maketitle
\begin{abstract}
This paper extends \citet{Becker1957}'s outcome test of discrimination to settings where a (human or algorithmic) decision-maker produces a ranked list of candidates. Ranked lists are particularly relevant in the context of online platforms that produce search results or feeds, and also arise when human decisionmakers express ordinal preferences over a list of candidates. We show that non-discrimination implies a system of moment inequalities, which intuitively impose that one cannot permute the position of a lower-ranked candidate from one group with a higher-ranked candidate from a second group and systematically improve the objective. Moreover, we show that that these moment inequalities are the \textit{only} testable implications of non-discrimination when the auditor observes only outcomes and group membership by rank. We show how to statistically test the implied inequalities, and validate our approach in an application using data from LinkedIn.\footnote{LinkedIn is engaged in a long-term effort to measure the impact of our products and ensure that they do not reinforce existing social inequalities. To learn more about these ongoing efforts, please visit: \href{https://lnkd.in/equitable-outcomes-meaning}{https://lnkd.in/equitable-outcomes-meaning}. Prior research showed that the InstaJobs algorithm can function to lower inequality broadly between different groups of members \citep{saint-jacques_fairness_2020}.}
\end{abstract}

\section{Introduction}

Researchers are often interested in testing whether a human or algorithmic decision-maker is biased against members of a protected group (e.g. race or gender). Substantial attention has been paid to the case where the decision is binary -- e.g., whether or not to grant a loan, accept a student to college, hire a job candidate, etc. However, in a variety of relevant domains, the decision-maker produces a \textit{ranked list} of candidates. Ranked lists are particularly relevant in the context of online platforms: LinkedIn provides recruiters with an ordered list of candidates, Google returns an ordered list of search results, and Facebook and Twitter provide users with an ordered feed of posts. Ranked lists are relevant in other domains, as well: for example, hospitals participating in the National Residency Match Program (NRMP) provide a listwise ranking of candidate residents \citep{roth_evolution_1984}, and experiments in behavioral economics have asked participants to rank which of the other participants they would like to be grouped with \citep{castillo_discrimination_2010}. 

We consider the setting where an Auditor observes ranked lists of candidates produced by a Ranker. The Auditor observes each candidate's group status $G$ and outcome $Y$. Importantly, the Auditor does not observe all of the features $X$ that are available to the decision-maker. This reflects a realistic limitation to (human or algorithmic) audits in a variety of domains. When the decision-maker is human, it is nearly always the case that there are factors that are observed by the decision-maker but not by the auditor -- for example, an auditor of the NRMP would not be able to observe everything that occurred during the candidate's interview. Likewise, external auditors of tech platforms will almost never have access to all of the features used by the algorithm. Even internally within tech platforms it is often difficult to retrospectively reconstruct all of the features used by an algorithm, since all of the relevant user data may not be saved for privacy reasons. Moreover, even if all the data were observed, the covariates may be so high-dimensional that it is difficult to condition on the full set of covariates in any practical analysis.

We then ask how the Auditor can test whether the Ranker is biased against a protected group in forming their rankings. Our notion of bias extends \citet{Becker1957}'s notion of tase-based discrimination to the context of listwise rankings. In particular, we will say that the Ranker is unbiased if they sort candidates to maximize an objective function that values placing candidates with better outcomes earlier in the list, regardless of their group status. As we show, this form of objective nests the Net Discounted Cumulative Gain (NDCG) objective commonly used for search algorithms. It can also be motivated by a simple model in which the objective is to maximize total engagement, and engagement with a post is an increasing function of its quality and rank in the list.

Our first main theoretical result is that the null hypothesis of no bias implies a system of moment inequalities. Intuitively, these moments impose that whenever we see a particular configuration of the candidates (e.g. a woman first, man second, etc.), we should not be able to flip the order of some of the candidates and improve the objective function. For instance, we should not be able to increase the objective on average by flipping the position of the first two candidates whenever the first-ranked candidate is male and the second is female. 
 
Our second main theoretical result is that this system of moment inequalities is a \textit{sharp} testable implication of the hypothesis of no bias. Specifically, whenever the moment inequalities are satisfied, there exists a distribution for the unobserved covariates such that the observed data corresponds with the utility maximization of an unbiased decision maker. 
 
Our theoretical results allow us to leverage a large econometrics literature on testing moment inequalities to develop statistical tests of bias in settings with list-wise rankings (see \citet{canay_practical_2017, molinari_chapter_2020} for reviews of the moment inequality literature). We discuss several practical considerations, including: reducing the dimension of the large number of implied inequalities; adjusting for position effects -- wherein a candidate's realized position has a causal effect on their outcome; and incorporating observed features about the candidates. 
 
Finally, we showcase our proposed procedure in a validation exercise using data from the InstaJobs algorithm at LinkedIn. The InstaJobs algorithm is an algorithm for determining whether to send users a notification about a job they may be interested in. The algorithm generates a predicted score for each candidate, and sends notifications to candidates above a threshold. We use the scores constructed by the algorithm for each job to create a listwise ranking of the candidates, and apply our proposed tests to this ranking data. This construction allows us to validate the listwise outcome test by comparing its results to what we would obtain by directly examining the score used to generate the ranking (which is not typically available for a ranking algorithm in practice). We find using the listwise outcome test that when a female and male candidate are in adjacent ranks, the \textit{lower-ranked} male candidate systematically has better outcomes than the \textit{higher-ranked} female candidates. This finding accords with a direct examination of the scores used for rankings, which show that the algorithm is under-calibrated for men relative to women.

Our results relate to a large literature on detecting discrimination in economics, computer science, and other fields. See \citet{lang_race_2020} for a recent review of the economics literature on discrimination, which has primarily focused on the case where the decision-maker makes individual-level decisions (e.g. give out a loan, hire a candidate) rather than produce listwise rankings.\footnote{One exception to this is \citet{castillo_discrimination_2010}, who conduct a lab experiment in which participants rank whom they would like to be grouped with. \citet{castillo_discrimination_2010} focus on differences in average ranks across groups.} Several notions of fairness for listwise ranking algorithms have been considered previously in the computer science literature, as reviewed in \citet{pitoura_fairness_2021}. Multiple papers have considered demographic parity constraints, which require that exposure (i.e. the distribution of rankings) be similar across groups \citep{zehlike_fair_2017,celis_ranking_2017,singh_fairness_2018,geyik_fairness-aware_2019, ghosh_when_2021}. Other work has considered the notions of disparate treatment and disparate impact, which restrict that exposure be proportional to average group-level utility or outcomes \citep{singh_fairness_2018}. \citet{beutel_fairness_2019} propose a notion of fairness that extends the notion of equal opportunity \citep{hardt_equality_2016} to the listwise ranking setting: this requires that the probability that a candidate is ranked below another candidate with a worse outcome does not differ across groups. The notion of fairness we consider here is distinct, and is based on the question of whether the ranking is consistent with maximizing an objective function that does not depend on group status directly. As discussed in \citet{corbett-davies_measure_2018} and \citet{rambachan_economic_2020} in the context of binary classification problems, decision rules that maximize an unbiased utility function may violate demographic parity or equalized odds if the distribution of risks differs across populations, a problem known as \textit{inframarginality}. Similar distinctions arise in the context of listwise rankings. 


\section{Model}

The model we consider consists of a Ranker and an Auditor. The Ranker -- which could be either an algorithm or human -- observes an unordered list of candidates and their characteristics, and produces a ranked list of the candidates. The Auditor then observes the ranked list of candidates and their outcomes, and wants to test whether the Ranker is biased. 
\subsection{Set-up}
\paragraph{Data-generating Process.} The Ranker is presented with queries indexed by $q$ in which they are asked to rank $J$ candidates with characteristics $X_{1q},...,X_{Jq}$ and group status $G_{1q},...,G_{Jq}$. We denote by $\mathcal{I}_q = \{ (X_{jq}, G_{jq} \}_{j=1}^J$ the information provided to the Ranker for query $q$. After observing $\mathcal{I}_q$, the Ranker produces a ranked list of the candidates -- formally this is a map $j_q: \{1,...,J\} \rightarrow \{1,...,J\}$ where $j_q(r)$ corresponds with the index of the candidate in rank $r$. Rank $1$ is the best rank, and we suppose that there are no ties in the rankings, so that $j_q$ is one-to-one. After the candidates are ranked, they realize outcomes $Y_{1q},...,Y_{Jq}$. We suppose for now that the outcomes $Y_{jq}$ do not depend on the rankings, and return to the scenario where outcomes may be affected by the rankings in Section \ref{sec: extensions}.

\paragraph{Auditor.} There is an Auditor tasked with evaluating whether the Ranker is biased. For each query $q$, the Auditor sees the rank-ordered list of outcomes $Y_q = (Y_{j_q(1)q},...,Y_{j_q(J)q})$ as well as the group statuses by rank, $G_q = (G_{j_q(1)q},...,G_{j_q(J)q})$. Importantly, the Auditor does not see the characteristics $X_q = (X_{j_q(1)q},...,X_{j_q(J)q})$ used to form the rankings. We show that similar results arise if the Auditor observes a subset of the variables in $X_q$ in Section \ref{sec: extensions}.

\subsection{Tests of Unbiasedness}

We now describe the notion of unbiasedness that we will test, which extends the logic of \citet{Becker1957} to the setting of list-wise rankings. Specifically, we will be interested in testing the hypothesis that the Ranker chooses the ranking $j_q(\cdot)$ to maximize 
\begin{equation} E\left[ \sum_r w_r Y_{j_q(r)q}  \,|\, \mathcal{I}_q \right], \label{eqn: expected weighted avg of Y} \end{equation}
where the $w_r$ are a strictly decreasing sequence of positive weights. Intuitively, the fact that the $w_r$ are decreasing means that the Ranker prefers to place candidates with higher values of $Y$ earlier in the list. This corresponds with expected utility maximization if the Ranker's utility function is $U(Y_q,G_q) = \sum_r w_r Y_{j_q(r)q}$, which depends only on the rank-ordered outcomes for the candidates, and not directly on their group-status $G_q$. 

\begin{defn}
We say that the Ranker is unbiased if they choose $j_q(\cdot)$ to maximize 
(\ref{eqn: expected weighted avg of Y}) for a decreasing sequence of positive weights $w_r$. 
\end{defn}

\noindent It is straightforward to show that an unbiased Ranker chooses the order $j_q(\cdot)$ that corresponds with sorting the candidates based on their expected outcomes given the Ranker's information set ($E[Y_{jq} | \mathcal{I}_q]$).

An important type of violation of this hypothesis is when the Ranker instead maximizes expected utility for the utility function
\begin{equation}
U(Y_q,G_q) = \sum_r w_r (Y_{j_q(r)q} - \tau G_{j_q(r)}) , \label{eqn: biased utility function}
\end{equation}
so that the outcome is effectively penalized by $\tau$ for candidates from group $G=1$ relative to group $G=0$. When the decision-maker is human, such a utility function may arise owing to racial animus against the $G=1$ group. This is thus a natural extension of \citet{Becker1957}'s notion of taste-based discrimination to the setting of listwise rankings. Although an algorithm will not typically have explicit racial animus, an algorithm may (approximately) maximize such an objective if there are biases in the training data so that the expected value of the outcome in the training data for $G=1$ is $\tau$ below its value in the target population. 

\begin{example}[NDCG] \label{example:ndcg}
A common objective function used for ranking algorithms is Net Discounted Cumulative Gain, abbreviated NDCG \citep{jarvelin_ir_2000}. Intuitively, NDCG is a weighted average of outcomes by position, normalized by the score that would be obtained if all candidates were sorted perfectly. Formally, NDCG is defined as the ratio $DCG/IDCG$, where 
$$DCG = \sum_r Y_{j_q(r)}^* / log_2(r+1),$$
\noindent $Y^*_j$ is a relevance score, and $IDCG$ is the idealized value of DCG that would be realized if the candidates had been sorted perfectly in decreasing order of $Y^*$,
$$IDCG = \sum_r Y_{j^*(r)}^* / log_2(r+1),$$
\noindent where $j^*(r)$ is the index with the $r$th largest value of $Y_j^*$. It is then apparent that maximizing expected NDCG is equivalent to maximizing an objective of the form (\ref{eqn: expected weighted avg of Y}), with $Y_j = Y_j^* / IDCG$, and $w_r = log_2(r+1)^{-1}$.

\end{example}

\begin{example}[Maximizing total engagement] \label{example:total_clicks}
It is well-known that users tend to engage more with posts earlier in a list than later in the list. Suppose that placing a post one position later in the list causally reduces the amount of user engagement by the factor $1+\gamma$. (The parameter $\gamma$ could be estimated in an experiment that randomizes the order of the candidates.) Then for $w_r = 1/(1+\gamma)^r$, we have that the total number of clicks generated by a given ranking is $\sum_r w_r Y_{j_q(r)q}$, where $Y_{jq}$ is the number of clicks that candidate $j$ would have received if placed in the first position in the query (and thus is not affected by rank).\footnote{If in practice candidate $j$ receives engagement $Y^*_{jq}$ and is placed in position $a$, we can construct $Y_{jq} = (1+\gamma)^a Y^*_{jq}$ as the ``position adjusted outcome.''} Thus, the objective function above corresponds with maximizing total engagement after accounting for the causal effect of position on engagement.
\end{example}

\section{Theoretical Results}

We now provide two main theoretical results. First, we show that unbiasedness by the Ranker implies a system of conditional moment inequalities. Second, we show that these moment inequalities are a sharp implication of unbiased behavior.

\begin{prop}\label{prop: moment inequalities}
If the Ranker is unbiased, then for all $a < b$,
\begin{equation}
E[ Y_{j_q(a)q} - Y_{j_q(b)q} \,|\, G_q = g] \geq 0    \label{eqn: moment inequalities}
\end{equation}
\noindent for all $g$ such that $P(G_q =g) > 0$. 
\end{prop}

\begin{proof}
Consider the counterfactual assignment rule that permutes $j_q(a)$ and $j_q(b)$ whenever $G_q = g$, and otherwise corresponds with the observed choice rule. Denote the rankings of this rule by $\tilde{j}_q(\cdot)$. Note that if $j_q(\cdot)$ maximizes $E\left[ \sum_r w_r Y_{j_q(r)q}  \,|\, \mathcal{I}_q \right]$ for all $\mathcal{I}_q$, then by iterated expectations it must maximize $E\left[ \sum_r w_r Y_{j_q(r)q} \right]$. Then, the difference in objective value from using the observed choice rule versus the permuted choice rule is given by
$$E\left[ \sum_r w_r Y_{j_q(r)q} \right] - E\left[ \sum_r w_r Y_{\tilde{j}_q(r)q} \right] =  P(G_q = g) (w_a - w_b) E[ Y_{j_q(a)q} - Y_{j_q(b)q} \,|\, G_q = g].$$

\noindent If the observed choice rule is optimal, then the expression in the previous display must be non-negative. However, $ P(G_q = g) (w_a - w_b) > 0$ by assumption, from which the result follows.
\end{proof}

\noindent It is straightforward to show that the change in the objective from permuting the candidates in positions $a$ and $b$ is proportional to $Y_{j_q(a)q} - Y_{j_q(b)q}$. Proposition \ref{prop: moment inequalities} thus intuitively states that if the Ranker is unbiased, then we shouldn't be able to permute the position of the candidates in positions $a$ and $b$ whenever we see $G_q = g$ and improve the objective. In other words, higher ranked candidates should always have higher values of $Y$ on average, regardless of the group orientation of the query.

\begin{example}
Suppose that $J=2$ and in every query there is one man and one woman, so that $G = (1,0)$ or $G = (0,1)$. Then (\ref{eqn: moment inequalities}) is equivalent to
\begin{align*}
& E[Y_{j_q(1)q} - Y_{j_q(2)q} \,|\, G_1 = 1, G_2 = 0] \geq 0 \\
& E[Y_{j_q(1)q} - Y_{j_q(2)q} \,|\, G_1 = 0, G_2 = 1] \geq 0
\end{align*}
\noindent The first inequality says that the outcome for the higher-ranked candidate should be larger on average when the higher-ranked candidate is male (and hence the lower-ranked candidate is female), whereas the second inequality is analogous for the case where the higher-ranked candidate is female.  
\end{example}

\begin{rem}[Sufficiency of adjacent ranks] \label{rem: can compare adjacent ranks}
Proposition \ref{eqn: moment inequalities} is stated in terms of comparisons of all ranks $(a,b)$ with $a<b$. It suffices to consider adjacent ranks, i.e. pairs of the form $(a,b) = (k, k+1$). This is because if the inequality in (\ref{eqn: moment inequalities}) holds for $(a,b) = (k,k+1)$ and $(a,b) = (k+1,k+2)$, then adding the two inequalities implies that it also holds for $(a,b) = (k,k+2)$, and so on.
\end{rem}
Our next result formalizes the notion that the inequalities in Proposition \ref{prop: moment inequalities} are the only testable implication of unbiasedness by the Ranker. It states that if the observed data satisfies the moment inequalities in Proposition \ref{prop: moment inequalities}, then there exists a latent distribution for the covariates $X_q$ such that the observed distribution corresponds with unbiased behavior by the decision-maker.

\begin{prop} \label{prop: moment inequalities are sharp}
Suppose the inequalities in Proposition \ref{prop: moment inequalities} are satisfied. Then there exists a distribution for $\mathcal{I}_q$ such that the observed distribution $(Y_q, G_q)$ corresponds with the decision rule that maximizes $E\left[ \sum_r w_r Y_{j_q(r)q} \,|\, \mathcal{I}_q \right]$.
\end{prop}
\begin{proof}
We construct a distribution for $\mathcal{I}_q$ that satisfies the proposition. Intuitively, we construct $\mathcal{I}_q$ such that whenever the decisionmaker chooses $G_q = g$, their expectation for the candidate in position $r$ is precisely $E[ Y_{j_q(r)q} \,|\, G_q =g ]$, and thus the observed ranking is optimal since this expectation is monotonically decreasing in the rank. Formally, let $\mathcal{I}_q^G = \{G_{jq}\}_{j=1}^{J}$ and $\mathcal{I}_q^X = \{X_{jq}\}_{j=1}^{J}$. Construct $\mathcal{I}_q^G$ to have the distribution corresponding with $\{G_q\}$, where $\{G_q\}$ denotes the unordered list of elements in $G_q$. Let $\mathcal{I}_q^X$ have support that is one-to-one with the support of $G_q$. In a slight abuse of notation, we will write $\mathcal{I}_q^X = g$ to denote that $\mathcal{I}_q^X$ takes the value in its support mapping to $g$. We construct $\mathcal{I}_q^X$ such that $P(\mathcal{I}_q^X = g \,|\, \mathcal{I}_q^G = \{g\}) = P(G_q = g \,|\, \{G_q\} = \{g\}). $ Next, we construct $Y_q \,|\, \mathcal{I}^X_q = g$ to have the same distribution as $Y_q \,|\, G_g = g$. Equation (\ref{eqn: moment inequalities}) then implies that if $r_1 < r_2$, then $E[Y_{j_q(r1)q} - Y_{j_q(r2)q} \,|\, \mathcal{I}_q^X = g] \geq 0$, and thus $j_q$ maximizes the objective. Moreover, by construction of the conditional probabilities, the implied distribution of $(Y_q, G_q)$ matches that in the data.
\end{proof}

\begin{rem}[Relationship to marginal outcome tests]
Unbiased behavior by the ranker implies that $E[Y_{j_q(a)} | \mathcal{I}_q]$ should be \textit{equal} to $E[Y_{j_q(a+1)} | \mathcal{I}_q]$ if the Ranker is indifferent between the candidates in positions $a$ and $a+1$ (marginal candidates). However, since it is not observed which candidates are marginal, the only testable implications involve comparisons between adjacent ranks, even when these are \textit{inframarginal}, meaning that the ranker strictly prefers the candidate in rank $a$ to the one in $a+1$. As a result, the inequality in (\ref{eqn: moment inequalities}) may be strict if the Ranker is unbiased -- that is, candidates ranked in position $a$ may be strictly better on average than candidates ranked in position $b>a$. By continuity arguments, it follows that if the Ranker has a small amount of bias (e.g. $\tau > 0$ is very small in (\ref{eqn: biased utility function})), then the inequalities in equation (\ref{eqn: moment inequalities}) might still be satisfied, so that discriminatory behavior is not detectable. However, the inframarginality problem should be less severe when many candidates are observed, since the expected difference between the $a$ and $(a+1)$-th best candidates should be small. For example, if $E[Y_{jq} | I_q]$ is $i.i.d.$ uniformly distributed across $j$, then for an unbiased ranker $E[Y_{j_q(a)q} - Y_{j_q(a+1)q} | I_q]$ is $O_P(1/J)$ uniformly in $a$.\footnote{This follows from the fact that the difference in consecutive order statistics of the uniform distribution is distributed $Beta(1,N)$.} 
\end{rem}

\section{Testing}

We now discuss how one can test the hypothesis that a Ranker is unbiased given a sample of queries $q=1,...,Q$. For simplicity, we will focus on the case where the queries $q$ are i.i.d., although the approach we describe will extend easily to clustered or weakly dependent data.

\paragraph{Pointwise tests.} We first note that for fixed values of $a,b$ and $g$, the hypothesis that equation (\ref{eqn: moment inequalities}) holds is simply the hypothesis that the population mean of $Y_{j_q(a)q} - Y_{j_q(b)q}$ is larger than zero among the population of queries with $G_q = g$. This individual hypothesis can be tested with a standard one-sided $t$-test for the mean of the population with $G_q = g$. Such tests will be asymptotically valid under standard regularity conditions that allow for an application of a central limit theorem. These individual tests for fixed values of $(a,b,g)$ are useful in that they can help identify where (if anywhere) the Decisionmaker appears to be making biased decisions, which may be useful in addressing any detected bias.

\paragraph{Joint tests.} Although the individual tests described above will be valid for each individual hypothesis for a fixed $(a,b,g)$ combination, it is well known that there is a problem of multiple hypothesis testing if such tests are conducted for each possible value of $(a,b,g)$. Fortunately, a large literature in econometrics has developed joint tests for a system of moment inequalities such as (\ref{eqn: expected weighted avg of Y}) for all relevant $(a,b,g)$; see, for instance, \citet{canay_practical_2017, molinari_chapter_2020} for recent reviews.

\subsection{Implementation and Extensions \label{sec: extensions}}
\paragraph{Choice of Moments.} One important practical point for implementation is that the dimensionality of the vector $G_q$ can be quite large in practice: for example, if $g$ takes on two values and there are 30 candidates in the query, then there are $2^{30} \approx 10^9$ possible values of $G_q$. To reduce the dimensionality in our implementation below, when comparing the outcomes of rank $a$ to rank $b$, we condition on the group membership of the candidates in ranks $a$ and $b$ (i.e. $G_{q,a},G_{q,b}$), but not on the group status of the other candidates in the list, which substantially reduces the dimensionality. Although in theory this may reduce the power of the test, we suspect that most of the pertinent information about the outcomes for ranks $a$ and $b$ is captured by their own group status, and not by the group status of other candidates in the query. In our implementation we also only test moments comparing adjacent ranks, which as discussed in Remark \ref{rem: can compare adjacent ranks}, is equivalent to the null hypothesis across all ranks.

\paragraph{Position effects.} Our analysis so far has assumed that the outcome for candidate $j$ does not depend on their rank in the query. In practice, however, there may be causal effects of position on the outcome, e.g. the same candidate may get more clicks if ranked higher in a search on LinkedIn. If the Auditor knows the causal effect of position -- e.g., from an experiment that randomizes search order -- the inequalities can be adjusted to account for this. In particular, suppose the Auditor knows that putting a candidate in position $a$ increases outcomes by a factor of $(1+\gamma)$ relative to ranking the same candidate in position $b$. Then the change in the objective from swapping the candidate in positions $a$ and $b$ would be proportional to $(1 + \gamma) Y_{j_q(b)q} - Y_{j_q(a)q}$, i.e. a comparison of the outcomes that each candidate would have reached if they had been placed in position $b$. Optimizing behavior by the Ranker implies that this swap can't improve the objective, and thus yields the inequality
$$E[ Y_{j_q(a)q} - (1+\gamma) Y_{j_q(b)q} \,|\, G_q = g] \geq 0$$
instead of (\ref{eqn: moment inequalities}). Note that this is equivalent to testing (\ref{eqn: moment inequalities}) where the outcome used is the ``position-adjusted'' outcome rather than the observed outcome, as in Example \ref{example:total_clicks}.

\paragraph{Partially observed features.} Our analysis so far has assumed that the features $X$ are completely unobserved by the auditor. In practice, a subset of the features used by the decisions-maker may be observed: that is, $X_q = (X^O_q, X^U_q)$, where $X^O_q$ are features observed by the auditor and $X^U_q$ are unobserved features. In this case, the same arguments above can be applied within each group of candidates with the same observable features -- i.e, conditional on $X^O_q$. Thus, the sharp testable implication of unbiased rankings in this case is that
$$E[ Y_{j_q(a)q} - Y_{j_q(b)q} \,|\, G_q = g, X^O_q = x] \geq 0$$ 
\noindent for almost-every $(g,x)$. Such inequalities can be tested using methods for \textit{conditional} moment inequalities \citep{andrews_inference_2013}. 

\section{Validation Using LinkedIn Data}

We now provide a validation exercise using data from LinkedIn's InstaJobs algorithm.\footnote{The version of this algorithm studied in this paper has subsequently been deprecated.}

\subsection{InstaJobs}

\paragraph{Background.} InstaJobs is an algorithm that sends LinkedIn members (candidates) a notification about a job posting that they may be interested in.\footnote{Here ``candidates'' refers to LinkedIn members who are candidates for receiving a notification regarding a job posting.} The algorithm uses features about the job posting and the candidates to predict the probability the candidate will apply for the job as well as the probability the application will receive attention from the recruiter. Specifically, the algorithm scores candidates based on the predicted value for the outcome
\begin{equation}
Y^*_{jq} = \alpha*1[job\ applied]_{jq}+ 1[job\ applied\ \&\ application\ received\ recruiter\ attention]_{jq} \label{eqn: instajobs objective}    
\end{equation}

\noindent where $\alpha \in (0,1)$ is a (proprietary) scalar parameter that determines the relative weight placed on applications versus recruiter attention. All candidates with a score above a certain threshold are sent a notification.

\paragraph{Creating Listwise Rankings.} Importantly, InstaJobs is not a listwise ranking algorithm, but rather a pointwise classification algorithm that creates a score and takes a binary decision based on this score. However, we can generate data \textit{as if} InstaJobs were a list-wise ranking algorithm by rank-ordering the candidates for each job by the score the InstaJob algorithm assigned them. An advantage of this approach is that, since we have access to the true scores underlying the InstaJobs algorithm, we can validate the listwise approach by investigating whether it produces similar answers to a direct investigation of the scores generating the rankings. Also, since InstaJobs is not actually a ranking algorithm, its results should not be subjected to potential position bias.

\paragraph{Practical Implementation.} We have data on the InstaJobs algorithm scores for approximately 193,000 jobs. We create a ranked list of the top 11 candidates for each job based on the InstaJobs score.\footnote{We restrict only to candidates who receive a notification, since outcomes are only available for these candidates. Thus, some queries will have fewer candidates if fewer than 11 people were notified for a given job.} (We use the top 11 candidates so that we can make 10 comparisons between adjacent ranks.) We then implement the listwise outcome test, using gender as the group variable and NDCG as the objective (see Example \ref{example:ndcg}). We construct tests for each individual comparison using one-sided $t$-tests, and joint tests based on moment inequality tests with least-favorable critical values that assume all moments have mean-zero (e.g. \citet{andrews_inference_2010}).\footnote{Formally, we create unconditional moments of the form $E\left[ \left(Y_{j_q(a)q} - Y_{j_q(a+1)q} \right) 1[G_{q,a} = g_1, G_{q,a+1} = g_2] \right] \geq 0 $, where $Y$ is the relevance score $Y^*$ normalized by the IDCG, as in Example \ref{example:ndcg}.} For comparison, we also can look directly at whether the score is calibrated to check whether marginally-notified candidates have the same outcomes regardless of their gender, where the outcome is the algorithm's objective function given in equation (\ref{eqn: instajobs objective}).

\paragraph{Results.} Figure \ref{fig:pointwise-moments} shows estimates of the moments comparing adjacent ranks. The left panel compares the average value of $Y$ for women in one rank relative to men in the adjacent rank below, with the x-axis denoting the rank of the woman. We see that the point estimate is negative for all ranks, indicating that the lower-ranked men actually have \textit{better} outcomes than the women ranked immediately above them, which means that we could improve the objective function (NDCG) by swapping their ranks. Moreover, the differences are individually statistically significant for 8 of the 10 ranks, and we can jointly reject the hypothesis that all of differences are non-negative ($p < 0.01$; see Table \ref{tbl:joint tests} for exact $p$-values). The right panel of Figure \ref{fig:pointwise-moments} makes similar comparisons, except between higher-ranked men and lower-ranked women, and finds only positive differences. In contrast to the results in the left panel, we see that men ranked above women do tend to have better outcomes. The listwise outcome test thus suggests that the algorithm is systematically ranking men below women despite having better outcomes. Since our data allows us to look at the scores generating the ranks, we can validate the listwise outcome test by looking at the outcomes by score directly: the binned scatter plot in Figure \ref{fig:pointwise-comparisons} shows that men with a given score do indeed have systematically higher outcomes than women with the same score, confirming the conclusion of the listwise outcome test.\footnote{This is true even if we impose a linear control for variation of scores within each score decile.}

\begin{figure}[!ht]
    \centering
    \includegraphics[width = 0.7\linewidth]{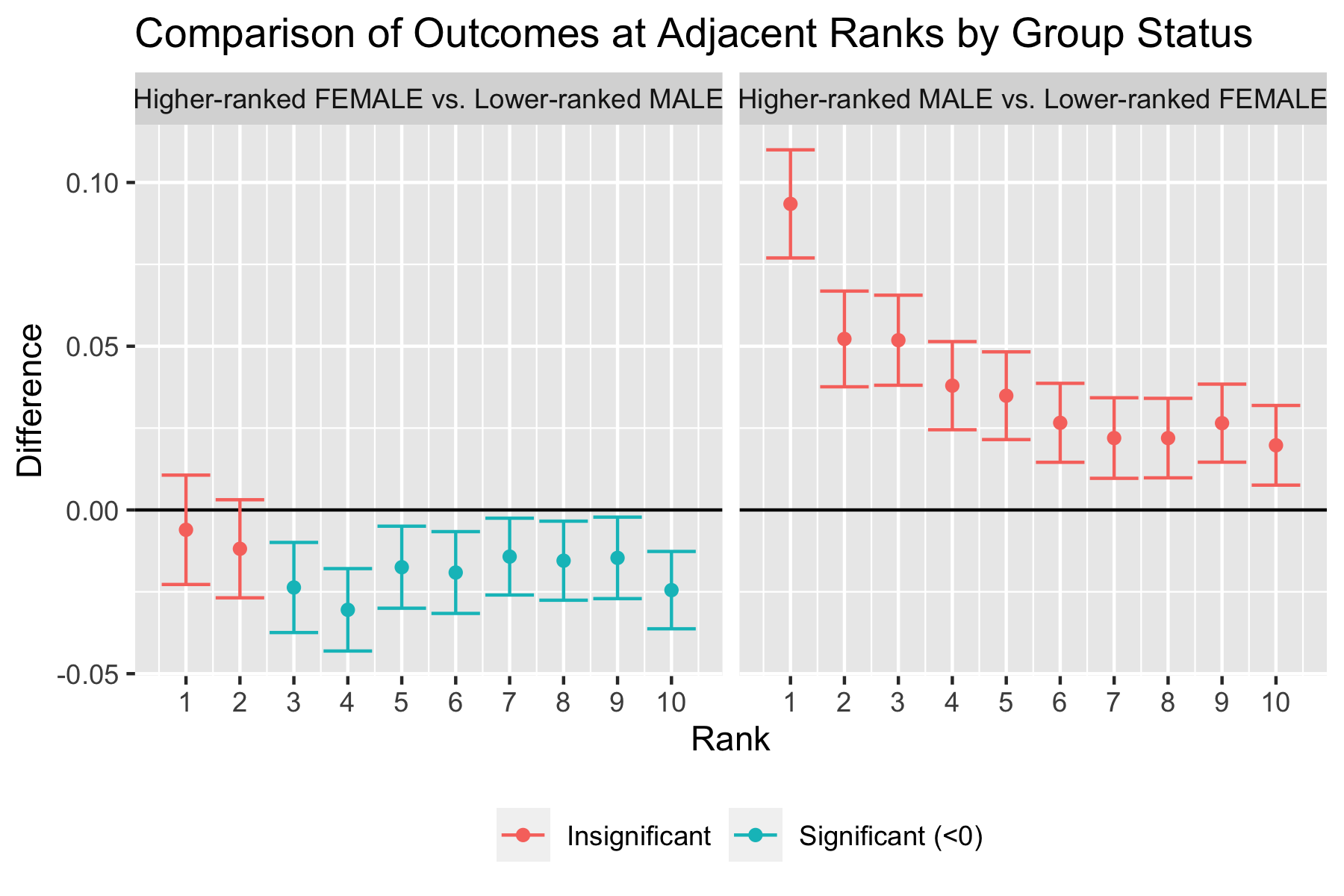}
    \caption{Moments Comparing Men and Women in Adjacent Ranks}
    \label{fig:pointwise-moments}
    \floatfoot{Note: This figure shows comparisons of the average outcome $Y$ between candidates in a given rank relative to the candidate in the next rank. The left panel shows these differences for queries where the higher-ranked candidate is female and the lower-ranked candidate is male. The right panel shows analogous comparisons when the higher-ranked candidate is male and the lower-ranked candidate is female. The $x$-axis denotes the rank of the higher-ranked candidate in the comparisons (where 1 is the highest rank).}
\end{figure}

\begin{table}[!hp]
    \includegraphics[]{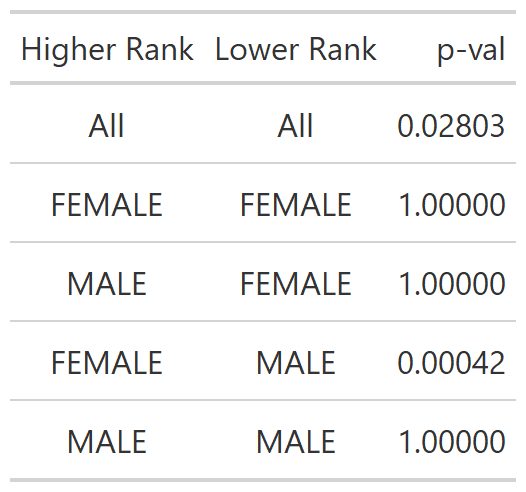}
    \caption{$p$-values for Joint Hypotheses}
    \label{tbl:joint tests}
    \floatfoot{Note: This table shows $p$-values for the joint hypothesis that all of the moments comparing higher-ranked to lower-ranked candidates are positive. The top row uses all of the moments, whereas the second only considers moments comparing women in adjacent ranks, the next row compares higher-ranked men to lower-ranked women, and so on. The $p$-values are constructed using least-favorable critical values for moment inequalities.}
\end{table}

\begin{figure}[!hbtp]
    \centering
    \includegraphics[width=0.7\linewidth]{Figures/pointwise-outcome-test.png}
    \caption{Pointwise Comparison of Outcome by Score}
    \label{fig:pointwise-comparisons}
    \floatfoot{Note: this figure shows a binned scatterplot of the outcome used by the InstaJobs algorithm against the InstaJobs algorithm score. The series are separated for male and female candidates.}
\end{figure}

\section{Conclusion}
This paper considers how to test for discrimination when the decision-maker (which could be a human or algorithm) produces a listwise ranking. We show that a sharp testable implication of unbiased behavior is a system of moment inequalities, and discuss how these can be tested in practice. We validate the methodology using the InstaJobs algorithm at LinkedIn. In future work, we plan to apply this methodology to other listwise rankings algorithms at LinkedIn.

\bibliography{Bibliography}


\counterwithin{figure}{section}
\counterwithin{table}{section}

\end{document}